\documentclass{article}
\usepackage[utf8]{inputenc} 

\usepackage{epsfig} 
\usepackage{amsmath} 
\usepackage{amssymb} 
\usepackage{amsthm}  
\usepackage{hyperref}
\usepackage{tikz}
\usetikzlibrary{automata,arrows,decorations.pathmorphing,decorations.markings,positioning,shapes,shapes.geometric,arrows.meta,bending,fit,calc}
\usepackage{scrextend}
\addtokomafont{labelinglabel}{\sffamily}
\usepackage{actuarialsymbol}
\usepackage[normalem]{ulem}
\usepackage{soul}

\usepackage{fullpage}
\usepackage{authblk}

\newtheorem{theorem}{Theorem}
\newtheorem{proposition}[theorem]{Proposition}

\newtheorem{lemma}[theorem]{Lemma}

\theoremstyle{definition}
\newtheorem{definition}[theorem]{Definition}

\newtheorem{example}[theorem]{Example}

\newtheorem{cor}[theorem]{Corollary}

\newcommand{\N}{{\mathbb{N}}}

\newcommand{\lex}{\mathit{lex}}
\newcommand{\rlex}{\mathit{co\text{-}lex}}
\newcommand{\colex}{\rlex}
\newcommand{\key}{\mathit{key}}
\newcommand{\first}{\mathit{first}}
\newcommand{\last}{\mathit{last}}
\newcommand{\suc}{\mathit{succ}}

\newcommand{\pmx}{{\sf{Pmax}}}
\newcommand{\pmn}{{\sf{Pmin}}}
\newcommand{\rpmx}{{\sf{rev\_Pmax}}}
\newcommand{\rpmn}{{\sf{rev\_Pmin}}}

\newcommand{\ignore}[1]{{}}

\title{A Cycle Joining Construction of the  Prefer-Max De Bruijn Sequence}

\author[1]{Gal Amram}
\author[2]{Amir Rubin}
\author[2]{Gera Weiss}
\date{}
\affil[1]{School of Computer Science, Tel Aviv University}
\affil[2]{Department of Computer Science, Ben-Gurion University of The
	Negev}

\begin{document}
	\maketitle
	\sloppy
	
	\abstract{We propose a novel construction for the well-known prefer-max De Bruijn sequence, based on the cycle joining technique. We further show that the construction implies known results from the literature in a straightforward manner. First, it implies the correctness of the onion theorem, stating that, effectively, the reverse of prefer-max is in fact an infinite De Bruijn sequence. Second, it implies the correctness of recently discovered shift rules for prefer-max, prefer-min, and their reversals. Lastly, it forms an alternative proof for the seminal FKM-theorem.}
	
	\section{Introduction}
	
	For $n>0$ and an alphabet $[k]=\{0,\dots,k{-}1\}$ where $k>0$, an $(n,k)$-De Bruijn sequence~\cite{bruijn1946combinatorial,Flye1894} ($(n,k)$-DB sequence, for abbreviation) is a total ordering of all words of length $n$ over $[k]$ such that the successor of each word is obtained by omitting its first symbol, and concatenating a final symbol to it instead. This applies also for the last and first words in the ordering. DB sequences admit applications in various fields including cryptography~\cite{AlhakimA11,AlhakimN17,DongP17,MykkeltveitS13,Turan11}, electrical
	engineering (mainly since they correspond to feedback-shift-registers)~\cite{Annexstein97,ChangTIHS1993,ChangELW18, lempel1970homomorphism,LiZLH14,LiZLHL16,MiroschnykKDKKF2018}, molecular biology~\cite{LinWJS11}, and neuroscience~\cite{AguirreMM11}, to name a few.
	
	The number of $(n,k)$-DB sequences (up to rotations) is $\frac{({k!}^{k^{n{-}1}})}{k^n}$~\cite{aardenne1951circuits}. Beyond this enumeration, many DB sequences were discovered for $k=2$ (e.g.~\cite{Alhakim10,DragonBH2016,Etzion1987,GabricS2017,Huang90,sawada2016surprisingly}). However, for the non-binary case, the number of known constructions is  smaller~\cite{AmramARSSW19,AmramRSW20,DragonHSWW18,FredricksenM1978,SAWADA2017524}.

	There are several standard construction techniques for DB sequences. First, a common technique addresses the notion of a De Bruijn graph (DB graph)~\cite{bruijn1946combinatorial,Good46}: the directed graph over all $n$-length words with edges of the form $((\sigma w),(w\tau))$. Clearly, DB sequences are the Hamiltonian cycles in the DB graph. Hence, generating Hamiltonian cycles in a DB graph forms a construction technique for DB sequences~\cite{AlhakimN17,BecherH11,LinWJS11,MiroschnykKDKKF2018,Turan11}.   
	Second, some DB sequences are proved to be equal to a concatenating of certain words~\cite{DragonHSWW18,FredricksenM1978}. Third, some DB sequences are constructed by a \emph{shift rule}; a formula to produce the successor of a given word~\cite{AmramARSSW19,AmramRSW20,DragonHSWW18,gabric2018framework,GabricSWW20,sawada2016surprisingly,SAWADA2017524}.
	
	This paper employs a fourth common construction technique: a \emph{cycle joining} (a.k.a. \emph{cross-join}) construction~\cite{AlhakimA11,AlhakimN17,Annexstein97,ChangTIHS1993,Fredricksen75,GabricSWW20,HaugeH96,JansenFB91,MykkeltveitS13,ZhangLH09}. Roughly speaking, in this technique, the $n$-length words are divided into the equivalence classes of the `rotation of' equivalence relation. Then, each equivalence class is ordered in a way that satisfies the successor property (hence, the classes are named cycles), and the cycles are  ``inserted" one into the other, so that the successor property holds for the entire obtained sequence.
	
	The most well-known and studied DB sequence is the prefer-max sequence~\cite{ford1957cyclic,martin1934problem}, and thus its complementary, the prefer-min sequence. These are the lexicographically maximal and minimal DB sequences. The prefer-max is generated by the ``granddaddy" greedy algorithm~\cite{martin1934problem} that repeatedly chooses the lexicographic maximal legal successor that have not been added yet. An analogous process produces the prefer-min sequence.
	
	Beyond this greedy approach, three constructions for the prefer-max sequence are known:
	\begin{enumerate}
		\item The seminal FKM theorem~\cite{FredricksenM1978} shows that prefer-max is a concatenation of certain Lyndon words~\cite{lyndon1954burnside}, i.e. non-periodic words that are lexicographically smaller among their rotations.
		
		\item A shift rule~\cite{AmramARSSW19,AmramRSW20,GabricSWW20} (a shift rule for $k=2$ is given in~\cite{fredricksen1972generation,weiss2007combinatorial}).
		
		\item A cycle joining construction~\cite{GabricSWW20}. 
	\end{enumerate}
	
	In~\cite{AmramARSSW19,GabricSWW20}, the shift rule was proved based on the correctness of the FKM theorem. In~\cite{AmramRSW20}, the other direction is established;~\cite{AmramRSW20} provides a self-contained correctness proof for the prefer-max shift rule, and show that the FKM theorem follows from it. Furthermore, both results, the FKM theorem, and the prefer-max shift rule, provide a cycle joining construction in a straightforward manner. To extract a cycle joining construction from a shift rule, one needs to identify the words whose successors/predecessors are not rotations of them. This is how a cycle joining construction was proved to generate prefer-max in~\cite{GabricSWW20}. A similar approach  can extract a  prefer-max cycle joining construction from the Lyndon words concatenation, i.e. the FKM theorem.
	
	This paper completes the missing ``triangle-edges". We provide a self-contained correctness proof for a cycle joining construction for prefer-max, see Section~\ref{sec:construction}. Then, we show that the prefer-max shift rule and thus the FKM theorem immediately follow from the correctness of our construction. Hence, we also provide an alternative proof for those two constructions of prefer-max, see Section~\ref{sec:implications}.
	
	To be precise, we provide a cycle joining construction for  the reverse of prefer-max. However, as we elaborate in the preliminaries, with reversing and complementing, the construction can be modified into constructions of prefer-max, prefer-min, and its reversal, see Section~\ref{sec:pre}. 
	
	Our construction follows a standard pattern. We order all cycles, and repeatedly insert the cycles into the sequence constructed so far. This approach differs from~\cite{GabricSWW20}, in which the set of all ``cycle-crossing" edges is identified. Note that that set can be extracted from the prefer-max shift rule. The pattern we follow allows us to easily conclude the \emph{onion theorem}~\cite{SchwartzSW19} (see also note in~\cite{BecherC21}). That is, the prefer-max over $[k{+}1]$ is a suffix of the prefer max over $[k]$ and hence, effectively, the reverse of prefer-max is an infinite DB sequence over the alphabet $\mathbb N$, see Section~\ref{sec:implications}.

	\section{Preliminaries}
	\label{sec:pre}
	
	We focus on alphabets of the form $\Sigma=[k]$ where
	$[k]=\{0,\dots,k{-}1\}$ and $k>0$, or $\Sigma=\mathbb N$. Naturally, the symbols in $\Sigma$ are totally ordered by $0<1<\cdots$.
	A word over $\Sigma$ is a sequence of symbols. Throughout the paper,  non-capital English letters denote words (e.g. $u,w,x,w'$ etc.), and non-capital Greek letters denote alphabet symbols (e.g. $\sigma,\tau,\sigma'$ etc.). $|w|$ denotes the length of a word $w$, and we say that $w$ is an $n$-word if $|w|=n$. $\varepsilon$ is the unique $0$-word. For a word $w$, $w^0=\varepsilon$, and $w^{t{+}1}=ww^t$.  $\Sigma^n$ is the set of all $n$-words over $\Sigma$. $R$ is the reverse operator over words, i.e. $R(\sigma_0\sigma_1\cdots\sigma_{n{-}2}\sigma_{n{-}1})=(\sigma_{n{-}1}\sigma_{n{-}2}\cdots\sigma_1\sigma_0)$.

	An $(n,k)$-DB sequence is a total ordering of $[k]^n$ satisfying: (1) a word of the form $\tau w$ is followed by a word of the form  $w\sigma$; (2) if the last word is $\tau x$, then the first word is of the form $x\sigma$.\footnote{ An $(n,k)$-DB sequence is also commonly defined as a (cyclic) sequence of $k^n$ symbols $\sigma_0\sigma_1\cdots$, for which every $n$-word appears in it as a subword. This presentation is essentially identical to ours by writing 
		$w_i =(\sigma_{i{-}(n{-}1)\mod k^n})\cdots (\sigma_{i{-}1\mod k^n}) \sigma_{i}$.
	}  We also consider infinite De Bruijn sequences. An $n$-DB sequence is a total ordering of  $\mathbb N^n$ that satisfies condition (1) above.
	
	For $n,k>0$, the prefer-max (resp. prefer-min)  $(n,k)$-DB sequence, denoted $\pmx(n,k)$ (resp. $\pmn(n,k)$), is the sequence constructed by the greedy algorithm that starts with $w_0=0^{n{-}1}{(k{-}1)}$ (resp. $(k{-}1)^{n{-}1}0$), and repeatedly chooses, while possible, the next word $w_i=x\tau$ where  $\tau$ is the maximal (resp. minimal) symbol such that property (1) holds, and $x\tau$ was not chosen in an earlier stage.
	\begin{example} 
		\label{exmp:prefmax}
		$\pmx(3,3)$ is the sequence:
		\begin{quote}
			$002,022, 222, 221, 212,  122,  220, 202, 021, 211, 112, 121,210, 102, 020,  201, \linebreak 012,  120, 200, 001, 011, 111,110,101, 010,100,000$.
		\end{quote}
		$\pmn(3,3)$ is the sequence:
		\begin{quote}
			$220,200, 000, 001, 010,  100,  002, 020, 201, 011, 110, 101,012, 120, 202,  021,\linebreak 210,  102, 022, 221, 211, 111,112,121, 212,122,222$.
		\end{quote}
		
	\end{example}
	Let $\rpmx(n,k)$ (resp. $\rpmn(n,k)$) be the reverse of $\pmx(n,k)$ (resp. $\pmn(n,k)$). That is, the sequence obtained by taking the $i$th word $u_i$ to be $R(w_{k^n{-}1{-}i})$, i.e. the reverse of the $k^n{-}1{-}i$ word of $\pmx(n,k)$ (resp. $\pmn(n,k)$).
	
	\begin{example} 
		\label{exmp:rev-prefmax}
		$\rpmx(3,3)$ is the sequence:
		\begin{quote}
			$000,001, 010, 101, 011,  111,  110, 100, 002, 021, 210, 102, 020, 201, 012,  121,\linebreak 211,  112, 120, 202, 022, 221, 212, 122, 222, 220, 200$.

		\end{quote}
		$\rpmn(3,3)$ is the sequence:
		\begin{quote}
			$222,221, 212, 121, 211,  111,  112, 122, 220, 201, 012, 120, 202, 021, 210,  101,\linebreak 011,  110, 102, 020, 200, 001, 010, 100, 000, 002, 022$.
		\end{quote}
	\end{example}
	
	Observe that $\pmx(n,k)$ and $\pmn(n,k)$ (and likewise $\rpmx(n,k)$ and $\rpmn(n,k)$) are derived one from the other by replacing each symbol $\sigma$ with $k{-}1{-}\sigma$. Hence, by reversing the sequence and subtracting the indices from $k{-}1$, properties and constructions for one of those sequences translate to corresponding properties and constructions for all other three.
	
	The fundamental FKM-theorem~\cite{fredricksen1977lexicographic,FredricksenM1978} links between the $\pmn$ sequence and Lyndon words~\cite{lyndon1954burnside}, as we elaborate below. A word $w'$ is a \emph{rotation} of $w$ if $w=xy$ and $w'=yx$. A word $w$ is \emph{periodic} if $w=x^t$ where $t>1$. A word  $w\neq \varepsilon$ is a \emph{Lyndon-word} if it is non-periodic, and lexicographically-minimal among its rotations. Let $\leq_\lex$ denote the lexicographic ordering of words.
	\begin{theorem}[FKM-Theorem.~Fredricksen, Kessler, Maiorana]
		\label{thm:FKM}
		Let $L_0{<_\lex} L_1{<_\lex}\cdots$ be all Lyndon-words over $[k]$ whose length divides $n$, ordered lexicographically. Then, $\pmn(n,k)=L_0L_1\cdots$.
		
	\end{theorem}
	
	As we present a construction for $\rpmx(n,k)$, the co-lexicographic ordering is of importance to us. That is, we write $w_1\leq_\rlex w_2$ if $w_1$ is a suffix of $w_2$, or we can write
	\[w_1=y_1\sigma x, \ w_2=y_2\tau x, \text{ such that }\sigma< \tau.\] Equivalently, $\colex$ can be defined by: $w_1\leq_\rlex w_2$ if $R(w_1)\leq_\lex R(w_2)$.

	\section{Cycle Joining Construction}
	\label{sec:construction}
	
	In this section we define the cycle joining construction for $\rpmx(n,k)$, and prove its correctness. We further conclude that $\rpmx$ is an infinite DB sequence. We obtain these results as follows. For every $n,k>0$ we construct an $(n,k)$-DB sequence, $D(n,k)$. We then show that our construction yields an $n$-DB sequence, $D(n)$, as $D(n,k)$ is a prefix of $D(n,k{+}1)$. Finally, we prove that $D(n,k)=\rpmx(n,k)$.
	
	\subsection{The Construction} 
	A key-word is an $n$-word that is $\colex$ maximal among its rotations.\footnote{This is inspired by the notion of Lyndon words - non-periodic words that are lexicographically minimal among their rotations. In our definition of a key-word we use $\colex$ ordering, take the maximal element, and do not require periodicity. } Let $\key_0,\key_1,\key_2,\dots$ be an enumeration of all key-words in $\rlex$ order. That is, $\key_0<_\rlex \key_1<_\rlex \cdots$. Let $c(n,k)$ be the number of all key-words of length $n$ over $[k]$.

	The cycle of $\key_m$ is a sequence $C_m$ whose elements are all $\key_m$-rotations, ordered as follows: If $m=0$, $C_m=(0^n)$. For $m>0$, let $\key_m=0^l(\sigma{+}1)w$. $w0^l(\sigma{+}1)$ is the first word in $C_m$, and each word $\sigma w'$ in $C_m$ is followed by $ w'\sigma$. Hence, the last word in $C_m$ is $(\sigma{+}1)w0^l$. We thus define the corresponding functions below. 
	\begin{definition}
		For $m>0$ let:
		\begin{itemize}
			\item $\key(C_m)= \key_m= 0^l(\sigma{+}1)w$;
			\item $\first(C_m)=w0^l(\sigma{+}1)$;
			\item $\last(C_m)= (\sigma{+}1)w0^l$.
		\end{itemize}
		Additionally, $\key(C_0)=\first(C_0)=\last(C_0)= 0^n$.
	\end{definition}
	
	\begin{example}
		For $n=6$, take $m$ such that $\key_m=002012$. Hence, $C_m=(\first(C_m)=012002,120020,200201,\key(C_m)=002012,020120, \last(C_m)=201200)$.
	\end{example}
	
	Note that $C_0,\dots,C_{c(n,k){-}1}$ partition the set $[k]^n$.
	
	We are ready to present the cycle joining construction. We inductively define an ordering of all words, where at step $m$ we extend the ordering of the elements in $\bigcup_{i=0}^{m{-}1} C_i$ to include the elements of $C_m$, {\color{blue} while} respecting the ordering within $C_m$ defined above.
	
	\begin{definition}[Construction Rule]
		\label{def:construction-rule}
		For each $m<c(n,k)$, we  inductively define a sequence $D_m$, over the words $\bigcup_{i=0}^m C_i$, as follows: 
		\begin{itemize}
			\item $D_0=(0^n)$.
			\item Write $\key_{m{+}1}=0^l(\sigma{+}1)w$. $D_{m{+}1}$ is obtained by inserting the sequence $C_{m{+}1}$ immediately after the word $\sigma w 0^l\in D_m$. 
			 Note that the key-word that is a rotation of $\sigma w 0^l$ is $\colex$ smaller than $0^l(\sigma{+}1)w$.
				Hence, indeed $\sigma w 0^l \in D_m$.
		\end{itemize}
		Let $D(n,k)=D_{c(n,k){-}1}$.
	\end{definition}

	\ignore{
	\begin{figure}
		\centering
		\begin{tikzpicture}[->,>=stealth',shorten >=1pt,auto,node distance=1.5cm,semithick,minimum size=0cm]

			\node[] (000) {$000$};
			\node[right = 15pt of 000] (001) {$001$};
			\node[right = 7pt of 001] (010) {$010$};
			\node[below right = 5pt and -8.5pt of 001] (100) {$100$};
			\node[right = 15pt of 010] (101) {$101$};
			\node[right = 7pt of 101] (011) {$011$};
			\node[below right = 5pt and -8.5pt of 101] (110) {$110$};
			\node[right = 15pt of 011] (111) {$111$};

			\draw [->,color=red] (000) edge  (001);
			\draw [->,color=red] (001) edge[bend left=75]  (010);
			\draw [->] (010) edge[bend left=40]  (100);

			\draw [->,color=red] (010) edge  (101);
			\draw [->,color=red] (101) edge[bend left=75]  (011);
			\draw [->] (011) edge[bend left=40]  (110);

			\draw [->,color=red] (110) edge[bend left=10]  (100);
			\draw [->,color=red] (011) edge  (111);
			\draw [->,color=red] (111) edge[bend left=30]  (110);

			\node[below right = 55pt and -5pt of 001] (002) {$002$};
			\node[right = 7pt of 002] (200) {$200$};
			\node[below right = 5pt and -8.5pt of 002] (020) {$020$};
			
			\node[left = 45pt of 002] (210) {$210$};
			\node[right = 7pt of 210] (021) {$021$};
			\node[below right = 5pt and -8.5pt of 210] (102) {$102$};
			
			\node[right = 15pt of 200] (222) {$222$};
			
			\node[right = 15pt of 222] (122) {$122$};
			\node[right = 7pt of 122] (212) {$212$};
			\node[below right = 5pt and -8.5pt of 122] (221) {$221$};

			\draw[->,color=blue] (100) edge[bend right = 20] (002);
			
			\draw[->] (002) edge[bend right =40] (020);
			\draw[->] (020) edge[bend right = 40] (200);
			\draw[->,color=blue] (002) edge (021);
			
			\draw[->,color=blue] (021) edge[bend right=75] (210);
			\draw[->,color=blue] (210) edge[bend right=40] (102);
			\draw[->,color=blue] (102) edge (020);

			\node[below = 55pt of 002] (012) {$012$};
			\node[right = 7pt of 012] (201) {$201$};
			\node[below right = 5pt and -8.5pt of 012] (120) {$120$};
			
			\node[below = 55pt of 210] (211) {$211$};
			\node[right = 7pt of 211] (121) {$121$};
			\node[below right = 5pt and -8.5pt of 211] (112) {$112$};

			\node[right = 15pt of 201] (220) {$220$};
			\node[right = 7pt of 220] (022) {$022$};
			\node[below right = 5pt and -8.5pt of 220] (202) {$202$};

			\draw[->,color=blue] (020) edge[bend left = 50] (201);
			\draw[->,color=blue] (201) edge[bend right=75] (012);
			\draw[->] (012) edge[bend right=40] (120);
			
			\draw[->,color=blue] (012) edge (121);
			\draw[->,color=blue] (121) edge[bend right =75] (211);
			\draw[->,color=blue] (211) edge[bend right = 40] (112);
			\draw[->,color=blue] (112) edge (120);
			
			\draw[->,color=blue] (120) edge (202);
			\draw[->,color=blue] (202) edge[bend right=40] (022);
			\draw[->] (022) edge[bend right = 40] (220);
			
			\draw[->,color=blue] (022) edge[bend right=10] (221);
			\draw[->,color=blue] (221) edge[bend right=40] (212);
			\draw[->,color=blue] (212) edge[bend right = 75] (122);
			
			\draw[->,color=blue] (122) edge (222);
			\draw[->,color=blue] (222) edge (220);
			\draw[->,color=blue] (220) edge[bend right =10] (200);

			
			\draw [-,dotted] (0.6,1.2) edge (0.6,-1.1);
			\draw [-,dotted] (3,1.2) edge (3,-1.1) ;
			\draw [-,dotted] (5.3,1.2) edge (5.3,-1.1) ;
			\draw [-,dotted] (-0.8,-1.1) edge (7,-1.1);
			
			\node[above = 10pt of 000] (c0) {$C_0$};
			\node[above = 29.3pt of 100] (c1) {$C_1$};
			\node[above = 29.3pt of 110] (c2) {$C_2$};
			\node[above = 10pt of 111] (c3) {$C_3$};

			
			\draw [-,dotted] (1.3,-1.2) edge (1.3,-3.5);
			\draw [-,dotted] (3.6,-1.2) edge (3.6,-3.5);
			\draw [-,dotted] (5,-1.2) edge (5,-3.5);

			\node[above = 29.3pt of 020] (c4) {$C_4$};
			\node[above = 29.3pt of 102] (c5) {$C_5$};
			\node[above = 29.3pt of 221] (c9) {$C_9$};
			\node[above = 10pt of 222] (c10) {$C_{10}$};
			
			
			\draw [-,dotted] (1.3,-3.6) edge (1.3,-6);
			\draw [-,dotted] (3.6,-3.6) edge (3.6,-6);
			\draw [-,dotted] (-0.8,-3.5) edge (7,-3.5);
			
			\node[above = 29.3pt of 120] (c6) {$C_6$};
			\node[above = 29.3pt of 112] (c7) {$C_7$};
			\node[above = 29.3pt of 202] (c8) {$C_8$};
			
		\end{tikzpicture}
		
		\caption{Our cycle joining construction. {\color{red}Red} arrows are $D(3,2)$ successors. {\color{blue}Blue} arrows extend $D(3,2)$ into $D(3,3)$. Black arrows are cycle-successors that are not $D(3,3)$ successors. }
		\label{fig:construction}
	\end{figure}
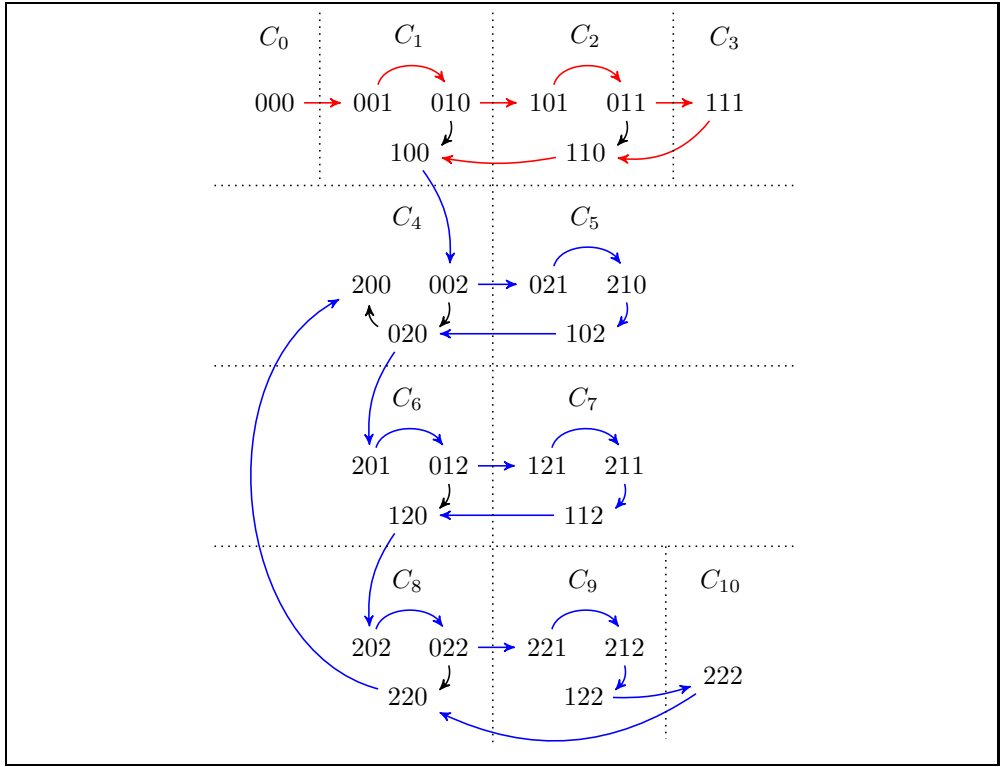}

	\begin{figure}
		\centering\framebox[0.8\textwidth]{
		\begin{tikzpicture}[->,>=stealth',shorten >=1pt,auto,node distance=1.5cm,semithick,minimum size=0cm]

			
			\node[] (000) {$000$};
			
			\node[right = 15pt of 000] (001) {$001$};
			\node[right = 7pt of 001] (010) {$010$};
			\node[below right = 5pt and -8.5pt of 001] (100) {$100$};
			
			\node[right = 15pt of 010] (101) {$101$};
			\node[right = 7pt of 101] (011) {$011$};
			\node[below right = 5pt and -8.5pt of 101] (110) {$110$};
			
			\node[right = 15pt of 011] (111) {$111$};
			
			\draw [->,color=red] (000) edge  (001);
			
			\draw [->,color=red] (001) edge[bend left=75]  (010);
			\draw [->] (010) edge[bend left=40]  (100);
			
			\draw [->,color=red] (010) edge  (101);
			
			\draw [->,color=red] (101) edge[bend left=75]  (011);
			\draw [->] (011) edge[bend left=40]  (110);
			
			\draw [->,color=red] (110) edge[bend left=10]  (100);
			
			\draw [->,color=red] (011) edge  (111);
			
			\draw [->,color=red] (111) edge[bend left=30]  (110);

			
			\node[below = 55pt of 101] (021) {$021$};
			\node[right = 7pt of 021] (210) {$210$};
			\node[below right = 5pt and -8.5pt of 021] (102) {$102$};
			
			\node[left = 45pt of 021] (200) {$200$};
			\node[right = 7pt of 200] (002) {$002$};
			\node[below right = 5pt and -8.5pt of 200] (020) {$020$};
			
			\draw[->,color=blue] (100) edge[bend left = 20] (002);
			
			\draw[->] (002) edge[bend left =40] (020);
			\draw[->] (020) edge[bend left =40] (200);
			
			\draw[->,color=blue] (002) edge (021);
			
			\draw[->,color=blue] (021) edge[bend left =75] (210);
			\draw[->,color=blue] (210) edge[bend left =40] (102);
			
			\draw[->,color=blue] (102) edge (020);

			
			\node[below = 55pt of 200] (201) {$201$};
			\node[right = 7pt of 201] (012) {$012$};
			\node[below right = 5pt and -8.5pt of 201] (120) {$120$};
			
			\node[right = 15pt of 012] (121) {$121$};
			\node[right = 7pt of 121] (211) {$211$};
			\node[below right = 5pt and -8.5pt of 121] (112) {$112$};
			
			\draw[->,color=blue] (020) edge[bend right = 20] (201);
			
			\draw[->,color=blue] (201) edge[bend left = 75] (012);
			\draw[->] (012) edge[bend left = 40] (120);
			
			\draw[->,color=blue] (012) edge (121);
			
			\draw[->,color=blue] (121) edge[bend left = 75] (211);
			\draw[->,color=blue] (211) edge[bend left =40] (112);
			
			\draw[->,color=blue] (112) edge (120);

			
			\node[below = 55pt of 201] (202) {$202$};
			\node[right = 7pt of 202] (022) {$022$};
			\node[below right = 5pt and -8.5pt of 202] (220) {$220$};
			
			\node[right = 15pt of 022] (221) {$221$};
			\node[right = 7pt of 221] (212) {$212$};
			\node[below right = 5pt and -8.5pt of 221] (122) {$122$};
			
			\node[below right = -3pt and 15pt of 212] (222) {$222$};
			
			\draw[->,color=blue] (120) edge[bend right =20] (202);
			
			\draw[->,color=blue] (202) edge[bend left = 75] (022);
			\draw[->] (022) edge[bend left = 40] (220);
			
			\draw[->,color=blue] (022) edge (221);
			
			\draw[->,color=blue] (221) edge[bend left = 75] (212);
			\draw[->,color=blue] (212) edge[bend left = 40] (122);
			
			\draw[->,color=blue] (122) edge[bend right = 10] (222);
			
			\draw[->,color=blue] (222) edge[bend left = 30] (220); 
			
			\draw[->,color=blue] (220) edge[bend left = 70] (200);
			
			
			\draw [-,dotted] (0.6,1.2) edge (0.6,-1.1);
			\draw [-,dotted] (2.9,1.2) edge (2.9,-8.6) ;
			\draw [-,dotted] (5.3,1.2) edge (5.3,-1.1) ;
			\draw [-,dotted] (-0.8,-1.1) edge (7,-1.1);
			
			\node[above = 10pt of 000] (c0) {$C_0$};
			\node[above = 29.3pt of 100] (c1) {$C_1$};
			\node[above = 29.3pt of 110] (c2) {$C_2$};
			\node[above = 10pt of 111] (c3) {$C_3$};

			
			\draw [-,dotted] (-0.8,-3.5) edge (7,-3.5);

			\node[above = 29.3pt of 020] (c4) {$C_4$};
			\node[above = 29.3pt of 102] (c5) {$C_5$};
			\node[above = 29.3pt of 122] (c9) {$C_9$};
			\node[right = 31pt of c9] (c10) {$C_{10}$};
			

			\draw [-,dotted] (-0.8,-5.9) edge (7,-5.9);
			
			\node[above = 29.3pt of 120] (c6) {$C_6$};
			\node[above = 29.3pt of 112] (c7) {$C_7$};
			\node[above = 29.3pt of 220] (c8) {$C_8$};
			
			
			\draw [-,dotted] (5.2,-5.9) edge (5.2,-8.5);

		\end{tikzpicture}
	}
		\caption{Our cycle joining construction. {\color{red}Red} arrows are $D(3,2)$ successors. {\color{blue}Blue} arrows extend $D(3,2)$ into $D(3,3)$. Black arrows are cycle-successors that are not $D(3,3)$ successors. }
		\label{fig:construction}
	\end{figure}

	\begin{example}
		\label{exm:construction}
		For $n=k=3$, the key-words are: 
		\begin{equation*}
			\begin{split}	
				& \key_0=000, \key_1=001, \key_2=011, \key_3=111, \key_4=002, \key_5=102, \\ & \key_6=012,  \key_7=112, \key_8=022, \key_9=122, \key_{10}=222.
			\end{split}
		\end{equation*}
		
		To demonstrate the cycle joining construction, we show $D(3,3)$ below, along with cycle-parenthesis. Parenthesis of cycle $C_i$ are denoted $(_i, \dots, \actsymb[i]{)}{}$. We also provide a construction illustration in Figure~\ref{fig:construction}.
		\begin{equation*}
			\begin{split}
				D(3,3)= & (_0 {\color{red}000}, \actsymb[0]{)}{}
				(_1  {\color{red}001},{\color{red}010}, 
				(_2 {\color{red}101}, {\color{red}011}, 
				(_3 {\color{red}111}, \actsymb[3]{)}{}  {\color{red}110}, \actsymb[2]{)}{}  {\color{red}100}, \actsymb[1]{)}{} 
				(_4  {\color{blue}002}, (_5 {\color{blue}021} ,
				\\ &  {\color{blue}210}, {\color{blue}102}, \actsymb[5]{)}{} {\color{blue}020},
				(_6 {\color{blue}201}, {\color{blue}012}, 
				(_7  {\color{blue}121}, {\color{blue}211},  {\color{blue}112}, \actsymb[7]{)}{} {\color{blue}120}, \actsymb[6]{)}{} 
				(_8 {\color{blue}202},  {\color{blue}022}, 
				\\ &  
				(_9 {\color{blue}221}, {\color{blue}212},  {\color{blue}122}, \actsymb[9]{)}{}
				(_{10} {\color{blue}222},  \actsymb[10]{)}{}  {\color{blue}220}, \actsymb[8]{)}{}   {\color{blue}200} \actsymb[4]{)}{}    .
			\end{split}
		\end{equation*}
		
	\end{example}
	
	Observe that the following two properties hold for $n=k=3$.
	\begin{enumerate}
		
		\item $D(3,2)$ (colored in red) is a prefix of $D(3,3)$.
		
		\item $D(3,3)=\rpmx(3,3)$ (see Example~\ref{exmp:prefmax}).
		
	\end{enumerate}
	Later, We will prove that these are general features of our cycle joining construction. We show now that we indeed construct a De Bruijn sequence.

	\begin{theorem}
		\label{thm:D(n,k)-DB}
		$D(n,k)$ is a De Bruijn sequence.
	\end{theorem}
	
	\begin{proof}
		First, we show that the successor property is an invariant of the construction as it holds for each sequence $D_m$. It vacuously holds for $D_0=(0^n)$. In addition, by the construction rule, if we insert a cycle $C_{m}$ between $\tau w$ and $w\sigma$, or after the final word $\tau w$,   then $\first(C_m)$ is of the form $w\xi$, and thus $\last(C_m)=\xi w$.
		
		Second, since at each step we insert cycles after some word, the first word in $D(n,k)$ is $0^n$. We need to show that the last word is of the form $\sigma 0^{n{-}1}$. We prove that this is an invariant of the construction as well. The claim vacuously holds for $D_0=(0^n)$. For the induction step, if we append a cycle $C$ to $D_m$, since $D_m$ ends with some $\sigma 0^{n{-}1}$, by the construction rule, $\first(C)=0^{n{-}1}(\sigma{+}1)$ and thus $\last(C)=(\sigma{+}1)0^{n{-}1}$. Hence, the invariant is preserved.  
	\end{proof}
	
	We also mention the next corollary that will use us later.
	
	\begin{cor}\label{cor:Dnk-ends-in-k-1}
		$D(n,k)$ ends in $(k{-}1)0^n$.
	\end{cor}
	\begin{proof}
		We showed that $D(n,k)$ ends in a word $\sigma 0^{n{-}1}$. By the construction rule, for each $\tau<k{-}1$, the cycle of $(\tau{+}1)0^{n{-}1}$ is inserted after $\tau0^{n{-}1}$. Therefore,  $(k{-}1)0^{n{-}1}$ is the final word of $D(n,k)$.
	\end{proof}
	

	From this point on, for  $n$-words $w,w'\in[k]^n$ we write $w< w'$ if $w$ appears before $w'$ in $D(n,k)$. Note that for each $k>0$, the linear ordering $D(n,k{+}1)$ extends $D(n,k)$. Hence, $<$ is well defined as it is not dependent on $k$. We further remark that the notation $<$ is reasonable since for the case $n=1$, $<$ coincides with the natural ordering of $[k]$

	\subsection{The Infinite Nature of the Construction}
	\label{sec:infinite}
	
	We turn to prove that our cycle joining construction  effectively provides us with an infinite De Bruijn sequence. Let $D=\bigcup_{k=0}^\infty D(n,k)$. That is, $D$ is a relation $D=(\mathbb N^n ,<)$, where $w_1< w_2$ if $w_1<w_2$ in some $D(n,k)$.

	\begin{theorem}
		\label{thm:D-infinite}
		$D$ is an infinite De Bruijn sequence.
	\end{theorem}
	
	Clearly, to prove Theorem~\ref{thm:D-infinite}, it is sufficient to show that each $D(n,k)$ is a prefix of $D(n,k{+}1)$.
	$D(n,k{+}1)$ is obtained by inserting the cycles $C_0,\dots,C_{c(n,k{+}1){-}1}$ one by one, into the sequence constructed so far. Towards proving our claim, we prove some progression property of the construction: we show that if $m<r$, then the cycle $C_r$ is not inserted before the cycle $C_m$. Later, this will use us to prove that $D(n,k)$ is a prefix of $D(n,k{+}1)$ as follows. We show that after we insert the cycles $C_0,\dots,C_{c(n,k){-}1}$ (and thus construct $D(n,k)$), the cycle $C_{c(n,k)}$ is appended to $D(n,k)$. As a result, by the progression property, we get the required. The next proposition formalizes the discussed progression property.
	
	\begin{proposition}
		\label{prop:first-in-order}
		If $m<r$, then $\first(C_m)<\first(C_r)$. 
	\end{proposition}	
	
	The remainder of this subsection is devoted for proving Proposition~\ref{prop:first-in-order}. We start with a few technical lemmas.

	\begin{lemma}
		\label{lem:decrease-by-one }
		If $0^l(\sigma{+}1)w$ is a key-word, then  $0^l\sigma w$ is a key-word.
	\end{lemma}
	\begin{proof}
		This lemma is essentially identical to~\cite[Lemma~7]{AmramARSSW19}, which deals with the analogous case for $\lex$-ordering, and minimal rotation.
	\end{proof}

	\begin{lemma}
		\label{lemma:finish-with-zeros}
		Let $0^l\sigma w\neq 0^n$ be a key-word, and $0< p\leq l$. Then, $0^p\sigma w0^{l{-}p}<\sigma w 0^l$.
	\end{lemma}
	\begin{proof}
		Write $0^l\sigma w=\key_m$ and consider the suffix of $C_m$, $(\key_m=0^l\sigma w,\dots, \last(C_m))$. In this subsequence, the trailing zeros on the left are shifted to the right one by one. Hence, $0^p\sigma w0^{l{-}p}$ appears before $\sigma w0^l$ in that sequence and thus $0^p\sigma w0^{l{-}p}<\sigma w 0^l$.
	\end{proof}

	\begin{lemma}
		\label{lemma:before-keyword}

				If $\key_m=0^l\sigma w$, and $\key_r=z\sigma w$ where $z\neq 0^{n{-}|\sigma w|}$, then all elements of $C_r$ precede $\sigma w0^l$ in $D$.

	\end{lemma} 
	\begin{proof}
		We prove by $\lex$-induction on $z$. That is, we assume that the lemma holds for each $z'<_\lex z$, and we prove for $z$.
		
		Note that $m<r$. We consider the sequence $D_{r}$, obtained by inserting $C_{r}$ into $D_{r{-}1}$ (which already includes the elements of $C_{m}$). As $z\neq 0^{l}$, we may write $\key_{r}=z\sigma w=0^t(\tau{+}1) y \sigma w$. By the construction rule, $C_{r}$ is inserted immediately after  $\tau y\sigma w0^t$. It is sufficient to prove that $\tau y \sigma w  0^t$ precede $\last(C_m)= \sigma w0^l$.
		
		First, note that, by Lemma~\ref{lem:decrease-by-one }, $0^t\tau y \sigma w$ is a key-word. If $0^t\tau y\neq 0^{l}$, we are done by the induction hypothesis. Otherwise, $\tau y \sigma w0^t= 0^{|\tau y|} \sigma w 0^t$, and we are done by the previous lemma.
	\end{proof}

	We can now prove Proposition~\ref{prop:first-in-order}.
	
	\begin{proof}[Proof of Proposition~\ref{prop:first-in-order}]
		Clearly, it suffices to prove for $r=m{+}1$.
		Write $\key_{m{+}1}=0^l(\sigma{+}1)w$. $C_{m{+}1}$ is inserted after $\sigma w 0^l$. By Lemma~\ref{lem:decrease-by-one }, $0^l\sigma w$ is a key-word. If $0^l\sigma w=\key_m$, we are done since we then have $\first(C_m)\leq \sigma w0^l<\first(C_{m{+}1})$. Otherwise, we have $0^l\sigma w<_\colex\key_m<_\colex\key_{m{+}1}=0^l(\sigma{+}1)w$. Therefore, $\key_m = z\sigma w$ where $z\neq 0^{l}$. Hence, by Lemma~\ref{lemma:before-keyword}, all elements of $C_m$ precede $\sigma w0^l$, and thus, in particular, $\first(C_m)<\first(C_{m{+}1})$.
	\end{proof}
	
	Finally, Theorem~\ref{thm:D-infinite} follows.
	
	\begin{proof}[Proof of Theorem~\ref{thm:D-infinite}]
		
		It is sufficient to show that $D(n,k)$ is a prefix of $D(n,k{+}1)$. By corollary~\ref{cor:Dnk-ends-in-k-1}, $D(n,k)=D_{c(n,k){-}1}$ ends in $(k{-}1)0^{n{-}1}$. $\key_{c(n,k)}$ is the $\colex$ minimal key-word in $[k{+}1]^n\setminus [k]^n$, hence $\key_{c(n,k)}=0^{n{-}1}k$. By the construction rule, $C_{c(n,k)}$ is inserted after $(k{-}1)0^{n{-}1}$. By Proposition~\ref{prop:first-in-order}, every cycle $C_m$, $m>c(n,k)$ is inserted after $\first(C_{c(n,k)})$ (and thus after $(k{-}1)0^{n{-}1}$). Consequently, $D(n,k)=D_{c(n,k){-}1}$ is a prefix of $D(n,k{+}1)=D_{c(n,k{+}1){-}1}$.
	\end{proof}

	\subsection{A Nesting Structure}
	\label{sec:containment}
	
	
	We next describe an interesting nesting structure that exists in our cycle joining construction. This structure will be used later to prove that our sequence is the reverse of the prefer-max sequence.
	
	We start with an immediate observation on the structure of $D_m$, implied by Proposition~\ref{prop:first-in-order}, which we call the \emph{parenthesis property}. If $m<r$, then, by Proposition~\ref{prop:first-in-order}, $C_r$ was not inserted before $C_m$. Therefore, either the cycle $C_r$ entirely follows $C_m$, or it is embedded into $C_m$. We term this property \emph{the parenthesis property} as we consider (virtual) parenthesis that wrap each cycle, as in Example~\ref{exm:construction}.  We turn to formalize the parenthesis property, and show when either of the cases holds.
	

	\begin{cor}
		[The Parenthesis Property]
		\label{cor:parenthesis}
		If $C_m$ and $C_r$ are two cycles such that $m<r$, then one of the following holds:
		\begin{enumerate}
			\item $\last(C_m)<\first(C_r)$, or
			
			\item $\first(C_m)<\first(C_r)\leq \last(C_r)<\last(C_m)$.
		\end{enumerate}
	\end{cor}

	
	\begin{definition}
		We say that $C_r$ is embedded in $C_m$, if $\first(C_m)<\first(C_r)\leq \last(C_r)<\last(C_m)$. 
		In addition, $C_r$ is said to be immediately embedded in $C_m$ if there is no cycle $C_l$ such that $C_r$ is embedded in $C_l$ and $C_l$ is embedded in $C_m$. 
		We inductively define the statement: ``$C_r$ is $t$-embedded in $C_m$":
		\begin{itemize}
			\item $C_r$ is $1$-embedded in $C_m$ if it is immediately embedded in $C_m$.
			\item $C_r$ is $(t{+}1)$-embedded in $C_m$ if there exists a cycle $C_l$ such that $C_r$ is $t$-embedded in $C_l$ and $C_l$ is $1$-embedded in $C_m$. 
		\end{itemize} 
	\end{definition}

	We now investigate relations between key-words of cycles $C_m$ and $C_r$, considering the two possible cases: When $C_r$ immediately follows $C_m$, and when it is immediately-embedded in $C_m$. First, we show that if we insert a cycle $C_r$ after $\last(C_m)$, then $\key_r$ is obtained by increasing the first non-zero symbol in $\key_m$ by one.

	\begin{lemma}
		\label{lemma:Cr-successor-of-Cm}
		Write  $\key_m=0^l(\sigma{+}1)w$.  If $\first(C_r)$ is the successor of $\last(C_m)$, then $\key_r=0^l(\sigma{+}2)w$.
	\end{lemma}   
	\begin{proof}
		Consider the sequence $D_r$, obtained by inserting $C_r$ into $D_{r{-}1}$. Write $\key_r=0^{l'}(\tau{+}1)w'$. By the construction rule, $C_r$ was inserted after $\tau w'0^{l'}$. By Proposition~\ref{prop:first-in-order}, $D_{r{-}1}$ already includes the elements of $C_m$ thus $\tau w'0^{l'}=\last(C_m)=(\sigma{+}1) w 0^l$. As $0^l(\sigma{+}1)w$ and $0^{l'}(\tau{+}1)w'$ are key-words, both $w$ and $w'$ end in a non-zero symbol. Hence, equality $\tau w'0^{l'}=(\sigma{+}1) w 0^l$ proves that $l=l'$, $\sigma{+}1=\tau$, and $w=w'$. As a result, $\key_r=0^{l'}(\tau+1)w'=0^l(\sigma{+}2)w$, as required.
	\end{proof}



	
	
	Now, we show that if we choose to embed $C_r$ in $C_m$, then $\key_m$ is obtained by zeroing the first non-zero symbol in $\key_r$.
	
	\begin{lemma}
		\label{embedding-properties}
		Assume that $C_r$ is immediately embedded in $C_m$. Write $\key(C_r)=0^i(\sigma{+}1)0^jw$ where $w$ does not start with $0$. Then,
		\begin{itemize}
			\item $\key(C_m) =0^{i{+}1{+}j}w$.
			\item If $u\in C_m$ and $\last(C_r)<u$, then $u=0^{j_2}w0^{i{+}1{+}j_1}$ where $j_1{+}j_2=j$. 
		\end{itemize}
	\end{lemma}

	\begin{proof}
		By Lemma~\ref{lem:decrease-by-one }, $0^{i{+}1{+}j}w$ is a key-word. Hence, to prove the first item,  we need to show that $0^{i{+}1{+}j}w\in C_m$. 
		
		Since $C_r$ is immediately embedded in $C_m$, the predecessor of $\first(C_r)$ is $v\in C_m$, or $\last(C)$, for some cycle $C$ that is also immediately embedded in $C_m$.  By repeatedly applying this reasoning, we construct a sequence of cycles 
		\[ C_{i_0},C_{i_1},\dots,C_{i_l}\text{ such that:}\]
		\begin{itemize}
			\item $C_{i_l}=C_r$.
			\item For each $ t\in\{1,\dots,l\}$, the predecessor of $\first(C_{i_t})$ is $\last(C_{i_{t{-}1}})$.
			\item The predecessor of $\first(C_{i_0})$ is a word $v\in C_m$.
		\end{itemize}
		
		Now, by applying Lemma~\ref{lemma:Cr-successor-of-Cm} $l$-times, $\key_{i_0}= 0^i(\tau{+}1)0^jw$, where $\sigma{+}1=\tau{+}1{+}l$.  Therefore, $\first(C_{i_0})=0^jw0^i(\tau{+}1)$, and, by the construction rule, 
		\[v=\tau 0^j w 0^i.\] 
		
		To prove that indeed $0^{i{+}1{+}j}w\in C_m$, we need to show that $\tau=0$ as $\tau=0$ implies that $0^{i{+}1{+}j}w$ is a rotation of $v\in C_m$. 
		 As $\key_{i_0}= 0^i(\tau{+}1)0^jw$ is a key-word, by Lemma~\ref{lem:decrease-by-one }, $0^i\tau 0^jw$ is also a key-word. As it is a rotation of $v\in C_m$,  $0^i\tau 0^jw=\key_m$. If $\tau>0$, then, by the definition of $last$, $v=\tau 0^jw0^i=\last(C_m)$, in contradiction to the fact that $C_r$ is embedded in $C_m$. Hence, $\tau=0$ and the first item holds. Moreover, the second item easily follows as $v,u,last(C_m)\in C_m$ and $$v=0^{j+1}w0^i<u\leq\last(C_m)=w0^{i+1+j}.\qedhere$$
		%
		%
		%
	\end{proof}

	By applying the previous lemma several times, we conclude the next corollary.
	
	\begin{cor}
		\label{cor:embedding-deletes-zeroes}
		Assume that $C_r$ is $t$-embedded in $C_m$. Write $\key(C_r)=uv$ where $u$ is the minimal prefix of $\key(C_r)$ that includes $t$ non-zero symbols. Then, $\key(C_m)=0^{|u|}v$.
	\end{cor}

	\subsection{Equivalence to the Reverse of the Prefer-Max}
	\label{sec:correctness}

	We are ready to show that we indeed construct the reverse of prefer-max.
	
	\begin{theorem}
		\label{thm:construction-correctness}
		$D(n,k)=\rpmx(n,k)$.
	\end{theorem}
	Recall that $\pmx(n,k)$ is the only DB sequence that (1) starts with $0^{n{-}1}(k{-}1)$ ,
	 and (2) $w(\tau{+}1)$ appears in it before $w\tau$ for every $w\in [k]^{n{-}1}$ and $\tau \in [k]$. Hence, to prove Theorem~\ref{thm:construction-correctness}, we shall prove the symmetric property: (1) $D(n,k)$ ends in $(k{-}1)0^{n-1}$, and (2) $\tau w<(\tau{+}1)w$. The former was already obtained in Corollary~\ref{cor:Dnk-ends-in-k-1}, and we focus on proving the later. 
	


	\begin{proposition}
		\label{prop:reverse-rule}
		For any word $\tau w$, $\tau w<(\tau{+}1)w$.
	\end{proposition}
	\begin{proof}
		Let $C_r$ be the cycle of $(\tau{+}1)w$. We start by proving the claim for the restricted case $(\tau{+}1)w=\last(C_r)$. Write $\key(C_r)=0^l(\sigma{+}1)w'$ and hence, $\last(C_r)=(\sigma{+}1)w'0^l=(\tau{+}1)w$, and $\first(C_r)=w'0^l(\sigma{+}1)$. By the construction rule, $C_r$ is inserted after $\sigma w'0^l=\tau w$, and the required follows.

		We turn to deal with the general case in which $(\tau{+}1)w\neq \last(C_r)$. Therefore, we may write $\key(C_r)=0^l(\sigma{+}1)w_1(\tau{+}1)w_2$, where 
		\begin{equation}
			\label{eq:not-last}
			(\tau{+}1)w=(\tau{+}1)w_20^l(\sigma{+}1)w_1.
		\end{equation}
\
			Let $C_m$ be the cycle of $\tau w$. Clearly, the maximal rotation of $\tau w$ is $\colex$ smaller than the maximal rotation of $(\tau{+}1)w$. Consequently, $\key_m<_\colex \key_r$ and thus $m<r$.
		Hence, by Proposition~\ref{prop:first-in-order}, 
		$\first(C_m)<\first(C_r)$.
		
		Now, if $\last(C_m)<\first(C_r)$, then every element of $C_m$ precedes every element of $C_r$ and we are done. Otherwise, by the parenthesis property, $C_r$ is embedded in $C_m$. For $\sigma\in[k]$,  let $|w|_\sigma$ denote the number of occurrences of  $\sigma$ in $w$, and note that $|\tau w|_0{-}|(\tau+1)w|_0\in\{0,1\}$. Use Corollary~\ref{cor:embedding-deletes-zeroes} to conclude that $|\tau w|_0{-}|(\tau{+}1)w|_0=1$ and that $C_r$ is immediately embedded in $C_m$. Moreover, as $|\tau w|_0{-}|(\tau{+}1)w|_0=1$, we have $\tau=0$. Hence, by Equation~\ref{eq:not-last},
		\begin{equation}
			\label{eq:tau-w-first-form}
			\tau w = 0w_20^l(\sigma{+}1)w_1.
		\end{equation}
		
		Furthermore, we get that the key of the cycle that includes $(\tau{+}1)w$, $\key(C_r)=0^l(\sigma{+}1)w_11w_2$. Write $w_1=0^jw_1'0^i$ and $w_2=0^p w_2'$ where $w_1'$ and $w_2'$ do not start or end with zero. Therefore, we have
		$\key(C_r)=0^l(\sigma{+}1)0^jw_1'0^i10^pw_2'$.
		
		
		Assume towards a contradiction that $1w=(\tau{+}1)w<\tau w=0w$. Recall that $\tau w=0w\in C_m$ and $C_r$ is embedded in $C_m$, and conclude (based on our assumption that $1w = (\tau {+} 1) w  < \tau w = 0w$) that the last element of $C_r$ must also appear before $0w$ : $1w \leq \last(C_r)<\tau w=0 w\leq \last(C_m)$. Therefore, by Lemma~\ref{embedding-properties}, 
		
		\begin{equation}
			\label{eq:j2<j}
			\tau w=0w=0^{j_2}w_1'0^i10^pw_2'0^{l{+}1{+}j_1}, \text{ where } 
			j_1{+}j_2=j.
		\end{equation}
		By Equation~\ref{eq:tau-w-first-form} and~\ref{eq:j2<j}, since $w_1=0^jw_1'0^i$ and $w_2=0^p w_2'$ ,we have:
		\begin{equation}
			\label{1}
			0^{j_2}w_1'0^i10^pw_2'0^{l{+}1{+}j_1}=0^{p{+}1}w_2'0^l(\sigma{+}1)0^jw_1'0^i.
		\end{equation}
		Therefore, $|0^{j_2}w_1'0^i10^rw_2'0^{l+1+j_1+1}|_1 = |0^{p+1}w_2'0^l(\sigma+1)0^jw_1'0^i|_1$ and thus $\sigma{+}1=1$. Hence,
		
		\begin{equation}
			\label{2}
			\key(C_r)=0^l10^jw_1'0^i10^pw_2'
		\end{equation}
		
		and Equation~\ref{1} can be rewritten as follows:
		
		\begin{equation}
			\label{3}
			0^{j_2}w_1'0^i10^pw_2'0^{l{+}1{+}j_1}=0^{p{+}1}w_2'0^l10^jw_1'0^i.
		\end{equation}
		
		For the remainder of the proof we assume that $w_1'\neq\varepsilon$ and $w_2'\neq \varepsilon$. The other cases are dealt similarly.

		By deleting the initial and final segments of zeros, we get from Equation \ref{3},
		\begin{equation}
			\label{4}
			j_2=p{+}1, \ \ w_1'0^i10^pw_2'=w_2'0^l10^jw_1'.
		\end{equation}
		Now, by Equation~\ref{2}, 
		\begin{equation}
			\label{5}
			0^i10^pw_2'0^l10^{j}w_1'\leq_{\rlex} 0^l10^jw_1'0^i10^pw_2'.
		\end{equation}
		
		By Equation~\ref{4}, these words have the same suffix  thus $0^i10^p \leq_{\rlex} 0^l10^j$.
		Hence, $j\leq p$. Therefore, by Equation~\ref{eq:j2<j}, $j_2\leq p$, in contradiction to Equation~\ref{4}.
	\end{proof}
	
	Finally, Theorem~\ref{thm:construction-correctness} follows.
	\begin{proof}[Proof of Theorem~\ref{thm:construction-correctness}]
		Since $D_0=0^n$, and since we insert cycles only after an existing element, $0^n$ is the first element of $D(n,k)$. $\rpmx (n,k)$ is the only sequence that includes all $n$-words (and no other elements), starts with $0^n$, and satisfies $\tau w<(\tau{+}1) w$ for each $\tau\leq k{-}2$ and $w\in[k]^{n{-}1}$. Hence, the theorem is implied by Proposition~\ref{prop:reverse-rule}.
	\end{proof}
	
	\section{Properties of Prefer-Max Implied by Our Construction}
	\label{sec:implications}
	
	We present applications induced by our construction. Specifically, first, we prove that $\rpmx$ is in fact an infinite De Bruijn sequence. Second, we extract from the construction the shift rule for \rpmx, proposed in~\cite{AmramRSW20}. Finally, as noted in~\cite{AmramRSW20}, this shift rule provides an alternative proof for the FKM-theorem.
	
	\subsection{The Onion Theorem}
	
	
	In Section~\ref{sec:infinite} we proved that $D(n)=\bigcup_{k=1}^\infty D(n,k)$ is an infinite De Bruijn sequence. In Section~\ref{sec:correctness} we proved that $D(n,k)=\rpmx(n,k)$. The onion-theorem~\cite{SchwartzSW19,BecherC21} follows.
	
	\begin{theorem}[The Onion-Theorem.]
		\label{thm:onion}
		$\rpmx(n)=\bigcup_{k=1}^\infty \rpmx(n,k)$ is an infinite De Bruijn sequence. 
	\end{theorem}

	\subsection{An Efficiently Computable Shift Rule.}
	
	By the correctness of our cycle joining construction, we conclude the correctness of the efficient shift rule given in~\cite{AmramRSW20}. For a word $w$, we write $\last(w)$ if $w=\last(C)$ for a cycle $C$. The successor function of $\rpmx(n,k)$ (resp. $\rpmx(n)$)  is ${\suc:[k]^n\setminus\{(k{-}1)0^{n{-}1}\}\rightarrow [k]^n}$ (resp. $\suc:\N^n\rightarrow \N^n$), defined by
	\[\suc(\sigma w)=\begin{cases} 
		w(\sigma{+}1) & \text{if }\last((\sigma{+}1)w) \\
		w0          & \text{if }\neg\last((\sigma{+}1)w) \text{ and }\last(\sigma w) \\
		w\sigma     & \text{otherwise} 
	\end{cases} .\]

	\begin{theorem}[Amram et al.~\cite{AmramRSW20}]
		\label{thm:rpmx-shift-rule}
		$\suc$ is a shift rule for $\rpmx(n,k)$.
	\end{theorem}
	\begin{proof}
		
		First, assume that $\last((\sigma{+}1)w)$. Let $C$ be the cycle of $(\sigma{+}1)w$, and note that $\first(C)=w(\sigma{+}1)$. Hence, by Definition~\ref{def:construction-rule},  $C$ was inserted after $\sigma w$. Furthermore, by the construction rule, no other cycle $C'$ was inserted after $\sigma w$ afterwards and thus the successor of $\sigma w$ is $w(\sigma{+}1)$.  
		
		Now, we handle the second case: $\neg\last((\sigma{+}1)w)$  and $\last(\sigma w)$. Write $\sigma w\in C_m$, and hence, $\first(C_m)=w\sigma$.
			Let $w\tau$ be the successor of $\sigma w$. First, we argue that $\neg\first(w\tau)$. Assume towards a contradiction that $\first(w\tau)$. Hence, by the construction rule, the cycle of $w\tau$ was inserted after $(\tau{-}1)w$. Use Proposition~\ref{prop:first-in-order} to conclude that no cycle was inserted between $(\tau{-}1)w$ and $w\tau$. Therefore, $\tau{-}1=\sigma$ and $\last (w (\sigma{+}1))$ follows, in contradiction to the assumption.  
			Now, since $\neg\first(w\tau)$,  the cycle of $\sigma w$ is immediately embedded in the cycle of $w\tau$. By Lemma~\ref{embedding-properties}, $\tau=0$ as required. 
		
		
		
		Lastly, we deal with the third case. Hence, $\neg\last(\sigma w)$ and thus the successor of $\sigma w$ in its cycle is $w\sigma$. Therefore, we should verify that no cycle $C_m$ was inserted between $\sigma w$ and $w\sigma$. Assume otherwise, and conclude that $\first(C_m)=w(\sigma{+}1)$. Hence, $\last(C_m)=(\sigma{+}1)w$, in contradiction to the case we are dealing with.  
	\end{proof}
	
	By Theorems~\ref{thm:onion} and ~\ref{thm:rpmx-shift-rule}, we conclude,
	
	\begin{theorem}
		$\suc$ is a shift rule for $\rpmx(n)$.
	\end{theorem}
	
	\subsection{The FKM Theorem}

	Following an observation from~\cite{AmramRSW20}, our results form an alternative proof for the seminal FKM-theorem (Theorem~\ref{thm:FKM}) as follows. For $n,k>0$, let $\mathit{next}:[k]^n\setminus \{0(k{-}1)^n\}\rightarrow[k]^n$ be the function constructed from $\mathit{succ}$ by the next rule: if $\mathit{succ}(\sigma_1\cdots\sigma_n)=\sigma_2\cdots\sigma_{n+1}$, then
	\[\mathit{next} ((k{-}1){-}\sigma_1\cdots,(k{-}1){-}\sigma_n)= ((k{-}1){-}\sigma_2\cdots, (k{-}1){-}\sigma_{n+1} )).\]
	Hence, $\mathit{next}$ is a shift rule for $\rpmn(n,k)$. Now, let $\mathit{next}^{-1}$ be the function constructed from $\mathit{next}$ by the next rule: if $\mathit{next}(\sigma_1\cdots\sigma_n)=\sigma_2\cdots\sigma_{n+1}$, then
	\[\mathit{next}^{-1}(\sigma_{n+1}\cdots\sigma_2)=\sigma_n\cdots\sigma_1.\]
	Hence, $\mathit{next}^{-1}$ is a shift rule for $\pmn(n,k)$. 
	We leave for the reader to verify that $\mathit{next}^{-1}$ is the shift rule proposed in~\cite{AmramARSSW19} (details can also be found in~\cite{AmramRSW20}).
	
	Now, let $L_0,L_1,\dots$ be an enumeration of all Lyndon words over $[k]$ whose length divides $n$, ordered lexicographically. Therefore, according to the proof of Theorem 4 in~\cite{AmramARSSW19}, $\mathit{next}^{-1}$ constructs the sequence $L_0L_1\cdots$, which implies that $\pmn(n,k)=L_0L_1\cdots$.

	\section{Conclusion}
	\label{sec:conclusion}
	
	For all $n,k>0$, we presented a cycle joining construction for the reverse of prefer-max sequence, $\rpmx(n,k)$. Since the sequences $\pmx(n,k)$, $\pmn(n,k)$, and $\rpmn(n,k)$ can be derived from $\rpmx(n,k)$, our construction can be modified into a cycle joining construction of any of those sequences.
	
	We showed that our construction implies the correctness of the \emph{onion-theorem}. That is, for all $n,k>0$, $\rpmx(n,k)$ is a prefix of $\rpmx(n,k{+}1)$, and thus $\rpmx(n)$ is an infinite DB sequence. Moreover, we showed that our construction also implies the correctness of the shift rules 
	given in $\cite{AmramRSW20}$. These shift rules are efficiently computable~\cite{AmramARSSW19, AmramRSW20}.

	As a result, our construction also implies the seminal FKM-theorem (Theorem~\ref{thm:FKM}). This theorem was presented in~\cite{FredricksenM1978} with only a partial proof: the described concatenation of  Lyndon words constructs a De Bruijn sequence. A quarter-century later, Moreno gave an alternative proof to that fact~\cite{moreno2004theorem}, and only a decade later, extended the proof, together with Perrin, into a complete proof for the FKM theorem~\cite{moreno2015corrigendum}. 
	Amram et al.~\cite{AmramRSW20} proved that the shift rule given in Section~\ref{sec:implications}, combined with statements proved in~\cite[Theorem~4]{AmramARSSW19} provide an alternative proof for Theorem~\ref{thm:FKM}. Hence, our cycle joining construction also constitutes an alternative proof for the FKM-theorem.

	\bibliographystyle{plain}
	\bibliography{bib}
	
\end{document}